\documentclass[a4paper,12pt]{article}
\usepackage{amscd}
\usepackage{amsmath,amsfonts,amssymb,amscd}
\usepackage{indentfirst,graphicx,epsfig}
\usepackage{graphicx,psfrag}
\input{epsf}

\newtheorem{thm}{Theorem}
\newtheorem{pro}{Proposition}
\newtheorem{cor}{Corollary}

\newtheorem{lem}{Lemma}

\newenvironment {proof} {\noindent{\em Proof.}}{\hspace*{\fill}$\Box$\par\vspace{4mm}}

\def\qed{\hfill \nopagebreak\rule{5pt}{8pt}}

\title{\bf Rainbow connections for planar graphs and line graphs\footnote{Supported by NSFC and ``the
Fundamental Research Funds for the Central Universities". } }
\author{
\small  Xiaolong Huang, Xueliang Li, Yongtang Shi\\
\small Center for Combinatorics and LPMC-TJKLC \\
\small Nankai University, Tianjin 300071, China \\
\small Email: xiaolonghuang@ymail.com, lxl@nankai.edu.cn,
shi@nankai.edu.cn
\date{}}
\begin{document}
\maketitle
\begin{abstract}
An edge-colored graph $G$ is rainbow connected if any two vertices
are connected by a path whose edges have distinct colors. The
rainbow connection number of a connected graph $G$, denoted by
$rc(G)$, is the smallest number of colors that are needed in order
to make $G$ rainbow connected. It was proved that computing $rc(G)$
is an NP-Hard problem, as well as that even deciding whether a graph
has $rc(G)=2$ is NP-Complete. It is known that deciding whether a
given edge-colored graph is rainbow connected is NP-Complete. We
will prove that it is still NP-Complete even when the edge-colored
graph is a planar bipartite graph. We also give upper bounds of the
rainbow connection number of outerplanar graphs with small
diameters. A vertex-colored graph is rainbow vertex-connected if any
two vertices are connected by a path whose internal vertices have
distinct colors. The rainbow vertex-connection number of a connected
graph $G$, denoted by $rvc(G)$, is the smallest number of colors
that are needed in order to make $G$ rainbow vertex-connected. It is
known that deciding whether a given vertex-colored graph is rainbow
vertex-connected is NP-Complete. We will prove that it is still
NP-Complete even when the vertex-colored
graph is a line graph. \\
[2mm] Keywords:  computational complexity; rainbow connection; coloring; planar graph; line graph\\
[2mm] AMS Subject Classification (2010): 68Q25, 68R10, 05C10, 05C12,
05C15, 05C76.
\end{abstract}

\section{Introduction}
All graphs considered here are simple, finite and undirected. We
follow the notation and terminology of \cite{West}. An edge-colored
graph is {\it rainbow connected} if any two vertices are connected
by a path whose edges have distinct colors (such paths are called
{\it rainbow path}). Obviously, if $G$ is rainbow connected, then it
is also connected. This concept of rainbow connection in graphs was
introduced by Chartrand et al. in \cite{CJMZ}. The rainbow
connection number of a connected graph $G$, denoted by $rc(G)$, is
the smallest number of colors that are needed in order to make $G$
rainbow connected. Observe that $diam(G)\leq rc(G)\leq n-1$. It is
easy to verify that $rc(G)=1$ if and only if $G$ is a complete
graph, that $rc(G)=n-1$ if and only if $G$ is a tree. Chartrand et
al. computed the precise rainbow connection number of several graph
classes including complete multipartite graphs (\cite{CJMZ}). The
rainbow connection number has been studied for further graph classes
in \cite{CLRTY, KS,LLiu, LS} and for graphs with fixed minimum
degree in \cite{CLRTY,KY,S}. There are also some results on the
aspect of extremal graph theory, such as \cite{S1}. Very recently,
many results on the rainbow connection have been reported in a
survey of Li and Sun \cite{LiSun}.

Besides its theoretical interest as a natural combinatorial concept,
rainbow connection has an interesting application for the secure
transfer of classified information between agencies (\cite{E}).
While the information needs to be protected, there must also be
procedures that permit access between appropriate parties. This
twofold issue can be addressed by assigning information transfer
paths between agencies which may have other agencies as
intermediaries, while requiring a large enough number of passwords
and firewalls that is prohibitive to intruders, yet small enough to
manage (that is, enough that one or more paths between every pair of
agencies have no password repeated). An immediate question arises:
what is the minimum number of passwords or firewalls needed that
allows one or more secure paths between every two agencies such that
the passwords along each path are distinct?

The complexity of determining the rainbow connection of a graph has
been studied in literature. It is proved that the computation of
$rc(G)$ is NP-hard \cite{CFMY, LZ}. In fact it is already
NP-complete to decide whether $rc(G)=2$, and in fact it is already
NP-complete to decide whether a given edge-colored (with an
unbounded number of colors) graph is rainbow connected \cite{CFMY}.
More generally it has been shown in \cite{LZ} that for any fixed
$k\geq 2$, deciding whether $rc(G)=k$ is NP-complete. Moreover, the
authors in \cite{LL} proved that it is still NP-Complete even when
the edge-colored graph is bipartite. Ananth and Nasre \cite{AN}
showed that for any fixed integer $k\geq 3$, deciding whether
$rc(G)=k$ is NP-Complete.

In this paper, we will prove that it is still NP-Complete to decide
whether a given edge-colored graph is rainbow connected even when
the edge-colored graph is a planar bipartite graph. As deciding
whether $rc(G)=2$ is NP-complete, the authors in \cite{DL, LLL}
considered bridgeless graphs with diameter two and proved that the
rainbow connection number in this case can not exceed $5$. We will
show that the rainbow connection number is at most three for
bridgeless outerplanar graphs with diameter two, and at most six for
bridgeless outerplanar graphs with diameter three.

A vertex-colored graph is {\it rainbow vertex-connected} if any two
vertices are connected by a path whose internal vertices have
distinct colors (such paths are called {\it vertex rainbow path}).
The {\it rainbow vertex-connection} of a connected graph $G$,
denoted by $rvc(G)$, is the smallest number of colors that are
needed in order to make $G$ rainbow vertex-connected. An easy
observation is that if $G$ is of order $n$ then $rvc(G)\leq n-2$ and
$rvc(G)=0$ if and only if $G$ is a complete graph. Notice that
$rvc(G)\geq diam(G)-1$ with equality if the diameter is $1$ or $2$.
For rainbow connection and rainbow vertex-connection, some examples
are given to show that there is no upper bound for one of parameters
in terms of the other in \cite{KY}. The rainbow vertex-connection
number has been studied for graphs with fixed minimum degree in
\cite{KY,LS1}. In \cite{CLS}, Chen, Li and Shi studied the
complexity of determining the rainbow vertex-connection of a graph
and prove that computing $rvc(G)$ is NP-Hard. Moreover, they proved
that it is already NP-Complete to decide whether $rvc(G) = 2$. They
also proved that it is already NP-complete to decide whether a given
vertex-colored graph is rainbow vertex-connected. In this paper, we
will prove that it is still NP-Complete to decide whether a given
vertex-colored graph is rainbow vertex-connected even when the
vertex-colored graph is a line graph.

\section{Rainbow connection for planar graphs}

Before proceeding, we list some related results as useful lemmas.

\begin{lem}[\cite{CFMY}]\label{lem1}
The following problem is NP-Complete: Given an edge-colored graph
$G$, check whether the given coloring makes $G$ rainbow connected.
\end{lem}

By subdividing each edge of a given edge-colored graph $G$ exactly
once, one can get a bipartite graph $G'$. Then color the edges of
$G'$ as follows: Let $e'$ and $e''$ be the two edges of $G'$
produced by subdividing at the edge $e$ of $G$. Then color the edge
$e'$ with the same color of $e$ and color the edge $e''$ with a new
color, such that all the new colors of the edges $e''$ are distinct.
In this way, Li and Li proved the following result from the problem
in Lemma \ref{lem1}.

\begin{lem}[\cite{LL}]\label{lem2}
Given an edge-colored bipartite graph $G$, checking whether the
given coloring makes $G$ rainbow connected is NP-Complete.
\end{lem}

A {\it plane graph} is a planar graph together with an embedding of
the graph in the plane. From the Jordan Closed Curve Theorem, we
know that a cycle $C$ in a plane graph separates the plane into two
regions, the interior of $C$ and the exterior of $C$. We prove the
following result.

\begin{thm}\label{thm1}
Given an edge-colored planar graph $G$, checking whether the given
coloring makes $G$ rainbow connected is NP-Complete.
\end{thm}
\begin{proof}
By Lemma \ref{lem1}, it will suffice by showing a polynomial
reduction from the problem in Lemma \ref{lem1}.

Given a graph $G=(V,E)$ and an edge-coloring $c$ of $G$, we will
construct an edge-colored planar graph $G'$ such that $G$ is rainbow
connected if and only if $G'$ is rainbow connected.

\begin{figure}[ht]
\begin{center}
\includegraphics[width=13cm]{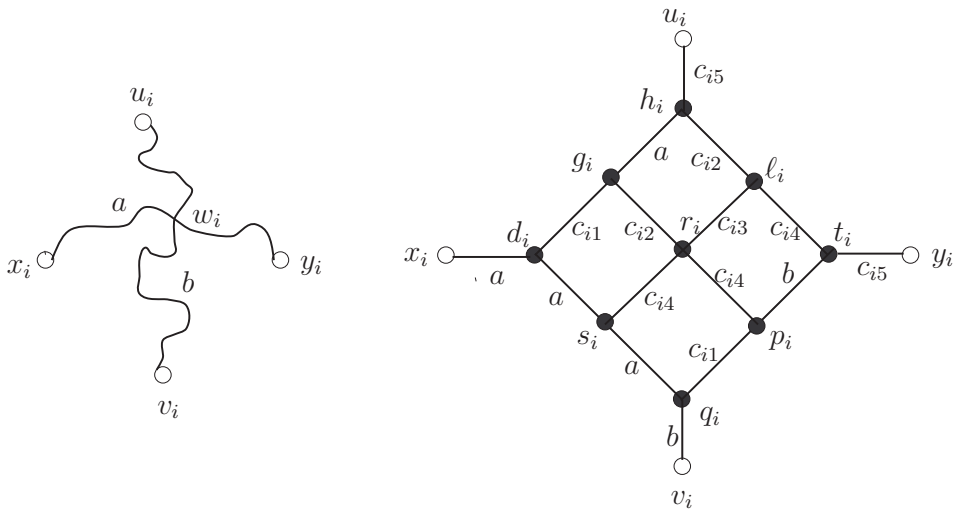}
\end{center}
\caption{The graph constructed in Theorem \ref{thm1} for some
crossing $w_i$.} \label{fig1}
\end{figure}

For one drawing of a given graph, by moving edges slightly, we can
ensure that no three edges have a common crossing and no two edges
cross more than once. Given a such drawing of $G$ in the plane with
$k$ crossings, denoted by $w_i$, where $i=1,2,\ldots,k$. Let $w_i$
be formed by two edges $x_iy_i$ and $u_iv_i$. First, we assume that
there is at most one crossing on each edge.

We construct an edge-colored graph $G'$ as follows. Graph
$G'=(V',E')$ is obtained from $G$ by replacing each crossing $w_i$
with one $3\times 3$-grid with vertex set $\{d_i,\ g_i,\ h_i, \
\ell_i,\ r_i,\ s_i,\ t_i,\ p_i,\ q_i\}$, as shown in Figure
\ref{fig1}. Therefore, we have $V'=V\cup \{d_i,\ g_i,\ h_i, \
\ell_i,\ r_i,\ s_i,\ t_i,\ p_i, q_i: 1\leq i\leq k \},$ $E'=E\cup
\{x_id_i,\ y_it_i,\ u_ih_i,\ v_iq_i,\ d_ig_i,\ g_ih_i,\ h_i\ell_i,\
g_ir_i,\ d_is_i,\ \ell_ir_i,\ r_is_i,\ \ell_it_i,\ r_ip_i,\ s_iq_i,\
p_iq_i,\ \\ p_it_i : 1\leq i\leq k\}$. From our construction, we
know that $G'$ is planar. In the following, we define an
edge-coloring $c'$ of $G'$: $c'(e)=c(e)$ for each $e\in E$;
$c'(x_id_i)=c'(d_is_i)=c'(s_iq_i)=c'(g_ih_i)=c(x_iy_i)$,
$c'(v_iq_i)=c'(p_it_i)=c(u_iv_i)$, $c'(d_ig_i)=c'(p_iq_i)=c_{i1}$,
$c'(g_ir_i)=c'(h_i\ell_i)=c_{i2}$, $c'(r_i\ell_i)=c_{i3}$,
$c'(r_ip_i)=c'(\ell_it_i)=c'(r_is_i)=c_{i4}$,
$c'(u_ih_i)=c'(t_iy_i)=c_{i5}$, where $c_{ij}$ are the new colors
for $1\leq i\leq k$ and $1\leq j\leq 5$.

Suppose coloring $c'$ makes $G'$ rainbow connected. For any two
vertices $u,v\in V$, there is a rainbow path $P'$ connected $u$ and
$v$. If $P'$ does not pass any grid, then $P'$ is also a rainbow
path joining $u$ and $v$ in $G$ under the coloring $c$. Otherwise,
suppose $P'$ passes some grid. We give the following claim.

{\bf Claim.} If the rainbow path $P'$ enters to a grid from vertex
$x_i$ (or $y_i$), then it must be go out from $y_i$ (or $x_i$).

Notice that $x_id_ig_ir_i\ell_it_iy_i$ is a rainbow path enters to
the grid from $x_i$ to $y_i$.  From the definition of $c'$, one can
easily to check that there has no rainbow path from $x_i$ (or $y_i$)
to $u_i$ and $v_i$, which just go through this grid. \qed

Similarly, one also can prove that if the rainbow path $P'$ enters
to a grid from vertex $u_i$ (or $v_i$), then it must be go out from
$v_i$ (or $u_i$). Denote by $P'(x_i,y_i)$ ($P'(u_i,v_i)$) the
subpath joining vertices $x_i$ and $y_i$ ($u_i$ and $v_i$) in path
$P'$ and let $P''$ be the path obtained from $P'$ by deleting
$P'(x_i,y_i)$ ($P'(u_i,v_i)$) and adding edge $x_iy_i$ ($u_iv_i$).
Applying this operation for each grid appeared in path $P'$ yields
one path $P$ of $G$, which is also a rainbow path in $G$ under the
coloring $c$. It follows that the coloring $c$ makes $G$ rainbow
connected.

To prove the other direction, suppose the coloring $c$ makes $G$
rainbow connected. Let $u$ and $v$ be a pair of vertices in $G'$. We
will find a rainbow path joining $u$ and $v$ in $G'$ under the
coloring $c'$ and then obtain that $c'$ makes $G$ rainbow
connection.

{\bf Case 1.} $u,v\in V$.

If there is a rainbow path joining $u$ and $v$ without going through
any crossing, then this path is also a rainbow path joining $u$ and
$v$ in $G'$ under the coloring $c'$. Now let $P$ be the rainbow path
joining $u$ and $v$ and some crossing $w_i$ lies on $P$.  Without
loss of generality, suppose $P=u\ldots x_iy_i\ldots v$. Then the new
path $P'$ obtained from $P$ by replacing the edge $x_iy_i$ with path
$x_id_ig_ir_i\ell_it_iy_i$ is the required rainbow path joining $u$
and $v$ in $G'$.

{\bf Case 2.} $u,v\in \{d_i,\ g_i,\ h_i, \ \ell_i,\ r_i,\ s_i,\
t_i,\ p_i,\ q_i: 1\leq i\leq k \}$, i.e., $u$ and $v$ belongs to the
same grid.

In this case, one can easily to find a rainbow path connecting $u$
and $v$ from the definition of $c'$.

{\bf Case 3.} $u\in V$, $v\in \{d_i,\ g_i,\ h_i, \ \ell_i,\ r_i,\
s_i,\ t_i,\ p_i, q_i: 1\leq i\leq k \}$.

It is easy to find the required rainbow path for the case of $u=u_i$
or $u=y_i$. Now suppose $u\notin \{u_i,v_i\}$. Since there exists a
rainbow path $P'$ joining $u$ and $u_i$ (or $y_i$) in $G'$ by {\bf
Case 1}, attaching the rainbow path between $u_i$ (or $y_i$) and $v$
to $P'$ yields the required rainbow path connecting $u$ and $v$.

{\bf Case 4.} $u$ and $v$ belongs to the different grids.

From the above cases, the proof of this case is obviously.

In any case, there exists one rainbow path connecting $u$ and $v$ in
$G'$ under the coloring $c'$.

Notice that this reduction is indeed a polynomial reduction, since
each graph has at most ${n\choose 2}$ crossings and for each
crossing, we introduce nine vertices, fourteen edges and five new
colors in the construction of graph $G'$.

Suppose there are more than one crossings on some edge $e$, we can
add one vertex with degree two between any two distinct crossings on
the same edge and then assign color $c(e)$ and a new color $c_1$ to
the two new edges. Since each graph has at most ${n\choose 2}$
crossings, we may introduce at most ${n\choose 2}$ new vertices and
${n\choose 2}$ new colors. Similarly, we can complete the polynomial
reduction.
\end{proof}

Using the same subdividing method for reducing Lemma \ref{lem1} to
Lemma \ref{lem2}, we can get the following corollary.

\begin{cor}
Given an edge-colored planar bipartite graph $G$, checking whether
the given coloring makes $G$ rainbow connected is NP-Complete.
\end{cor}

We now consider a more restricted class of planar graphs, namely
outerplanar graphs. A planar graph $G$ is said to be {\it
outerplanar} if $G$ can be embedded in the plane in such a way that
all vertices are incident with a common face. From this it is easy
to see that any $2$-connected outerplanar graph has a Hamilton
cycle. Now we give a property of outerplanar graphs.

\begin{pro}[\cite{West}]\label{proposition1}
Every simple outerplanar graph has a vertex of degree at most two.
\end{pro}

In \cite{LLiu}, the authors proved that $rc(G)\leq \lceil n/2\rceil$
for any $2$-connected graph $G$. Since $rc(C_n)\leq \lceil
n/2\rceil$, where $C_n$ denotes the cycle graph of order $n$, we can
deduce that $rc(G)\leq \lceil n/2\rceil$ for any Hamiltonian graph
$G$.

\begin{pro}
Let $G$ be a Hamiltonian graph, then $rc(G)\leq \lceil n/2\rceil$.
\end{pro}

As the rainbow connection number is at least the diameter and
deciding whether $rc(G)=2$ is NP-complete, it is necessary to
determine the rainbow connection number of graphs with diameter two.
The authors in \cite{DL, LLL} considered bridgeless graphs with
diameter two and proved that the rainbow connection number in this
case can not exceed $5$.

\begin{lem}[\cite{DL, LLL}]
If G is a connected bridgeless graph with diameter $2$, then
$rc(G)\leq 5$. Moreover, the upper bound is sharp.
\end{lem}

We show that the rainbow connection number is at most three for
bridgeless outerplanar graphs with diameter two, and at most six for
bridgeless outerplanar graphs with diameter three.

A subset $D$ of the vertices in $G$ is called a {\it dominating set}
if every vertex of $G-D$ is adjacent to a vertex of $D$.
Furthermore, if the dominating set $D$ induces a connect subgraph of
$G$, then $D$ is called a {\it connected dominating set}. Let
$X,Y\in V(G)$, we say that $X$ {\it dominates} $Y$ if every vertex
of $Y$ is adjacent to at least one vertex of $X$. The following
lemma will be used in the sequent.
\begin{lem}[\cite{CDRV}]
For any connected graph $G$ with minimum degree at least two. Let
$D$ be a connected dominating set of $G$, then $rc(G)\leq
rc(G[D])+3$.
\end{lem}

\begin{thm}
If $G$ is a bridgeless outerplanar graph with order $n$ and diameter
two, then $rc(G)\leq 3$, i.e., $rc(G)=2,3$.
\end{thm}
\begin{proof}
Suppose that $G$ is a bridgeless outerplanar graph with diameter
two. If $G$ has a cut vertex, then this vertex is a domination set
of the graph, then $rc(G)\leq 3$. Now we suppose that $G$ is
$2$-connected and we can embed $G$ so that a Hamilton cycle, $H$,
bounds the outer face, and the edges not in $H$ are chords that lie
in the interior of $H$. If $G$ has no chords, then $G$ is a cycle of
length at most five and thus $rc(G)\leq 3$. In the following we
assume $G$ has chords. Let $v$ be a vertex with degree two and
suppose $N(v)=\{x_1,y_1\}$. Denote by $C$ the induced cycle of $G$
containing vertex $v$. We will consider the following two cases
according to the order of $C$.

{\bf Case 1.} $|C|=4$.

Suppose $C=vx_1zy_1v$. In this case, there are at most two vertices
outside of $C$, since each vertex outside of $C$ must be adjacent to
both $x_1$ and $z$ (or $y_1$ and $z$). Observe that $rc(G)=2$.

{\bf Case 2.} $|C|=3$.

For convenience, we assume $H=vx_1x_2\ldots x_{n/2}y_{(n-2)/2}
\ldots y_2y_1v$ for even $n$ and $H=vx_1x_2\ldots
x_{(n-1)/2}y_{(n-1)/2} \ldots y_2y_1v$ for odd $n$. If $H$ has only
one chord, then this case is the same as {\bf Case 1}. Otherwise,
$H$ has at least two chords and then $n\geq 5$. There must be one
chord $e$ such that one of its end vertex is $x_1$ or $y_1$, without
loss of generality, say $x_1$. Then, the other end of $e$ must be
$y_2$ or $x_3$. Assume $e=x_1y_2$, then all other vertices in the
set $V\setminus \{v,x_1,x_2,y_1,y_2\}$ must be adjacent to $x_1$, as
the diameter of $G$ is two. Therefore, in this case, the structure
of graph $G$ is a fan, as shown in Figure \ref{fig2}.

\begin{figure}[ht]
\begin{center}
\includegraphics[width=8cm]{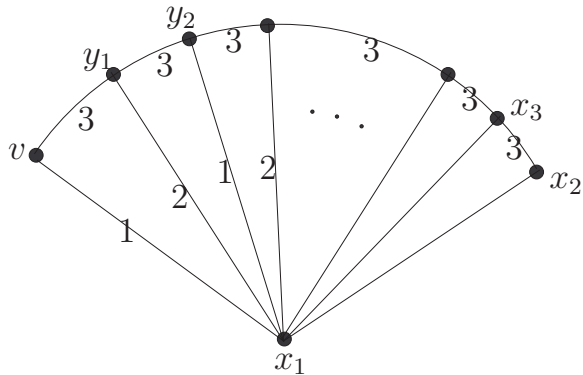}
\end{center}
\caption{Bridgeless outerplanar graph with diameter
two.}\label{fig2}
\end{figure}

For $n=5$, we can give an edge coloring $c$ of $G$ such that
$rc(G)=2$ under this coloring:
$c(vx_1)=c(vy_1)=c(x_1y_2)=c(x_2y_2)=1$ and
$c(x_1y_1)=c(y_1y_2)=c(x_1x_2)=2$. For $n\geq 6$, we observe that
$rc(G)=3$. Notice that $2$ colors cannot make $G$ rainbow connected.
Now we give one edge coloring with three colors: all edges with
$x_1$ as one of its end are assigned colors $1$ and $2$
alternatively in clockwise order; all other edges are assigned color
$3$, as shown in Figure \ref{fig2}.
\end{proof}

\begin{thm}
If $G$ is a bridgeless outerplanar graph with order $n$ and diameter
three, then $3\leq rc(G)\leq 6$.
\end{thm}
\begin{proof}
Suppose that $G=(V,E)$ is a bridgeless outerplanar graph with
diameter three. Since the rainbow connection number is at least the
diameter, then we have $rc(G)\geq 3$. Suppose $G$ is not
$2$-connected and let $v$ be a cut vertex of $G$. There is a
partition of $V-\{v\}$ into two sets $A$ and $B$ such that vertex
$v$ dominates either $A$ or $B$. Without loss of generality, we may
assume $v$ dominates $A$. Denote by $B_1$ the vertices of $B$ that
are adjacent to $v$ and $B_2=B-B_1$. Choose a minimum cardinality
subset $S$ of $B_1$ such that $S$ dominates $B_2$. Then $S\cup
\{v\}$ is a connected dominating set. We claim that $|S|\leq 2$.
Suppose that $|S|\geq 3$ and let $S=\{s_1,s_2,s_3\}$. By the
minimality of $S$, there exist three vertices $x_1,x_2,x_3\in B_2$
satisfying that among three vertices $s_1,s_2,s_3$, $x_i$ is only
adjacent to vertex $s_i$ for $1\leq i\leq 3$. Take $a\in A$. Without
loss of generality, we may assume that an embedding of $G$ as an
outerplanar graph has vertices $a,s_1,s_2,s_3$ in clockwise order,
adjacent to $v$. Since all the vertices of $G$ lie on a common face,
there is no way to obtain a path of length at most three between
$x_1$ and $x_3$. Thus, $|S|\leq 2$, which yields that $rc(G)\leq
rc(G[S\cup \{v\}])+3=5$.

Now suppose $G$ is $2$-connected. It follows that $G$ can be
embedded in such a way that a Hamilton cycle $H$ bounds the outer
face, and the edges not in $H$ are chords that lie in the interior
of $H$. If $H$ has no chords, then $G$ is a cycle of length at most
seven, and thus $rc(G)=3$ or $rc(G)=4$. Thus, in the following we
assume $H$ has at least one chord.

Suppose $xy$ is a chord of $H$. Cycle $H$ is divided into two
$xy$-paths. We denote the path goes in clockwise direction from $x$
to $y$ by the $xy$-segment of $H$, and denote the other path by the
$yx$-segment of $H$.

Now suppose $H$ has precisely one chord $xy$. In this case,
$\{x,y\}$ is a vertex cut of $G$. Since $G$ has diameter three, then
$\{x,y\}$ dominates either $xy$-segment of $H$ or the $yx$-segment
of $H$. Without loss of generality, we suppose that $xy$-segment is
dominated. Since there are no other chords, the $xy$-segment of $H$
is a path of length two or three. If it is two, then the
$yx$-segment of $H$ is a path of length four or five and thus we can
check that $rc(G)=3$. Otherwise, the $yx$-segment of $H$ is a path
of length three or four and thus $rc(G)=3$ or $4$.

Suppose $H$ has at least two chords. Among all vertex cuts with two
vertices, we choose $\{a,b\}$ as a vertex cut such that it dominates
a maximum number of vertices. Note that $a$ and $b$ may not
correspond to the ends of a chord of $H$. Since $G$ has diameter
three, $\{a,b\}$ dominates one segment of $H$. Without loss of
generality, we assume $ba$-segment of $H$ is dominated by $\{a,b\}$.
Consider the $ab$-segment of $H$.

{\bf Case 1.} There are no chords with both ends on the $ab$-segment
of $H$.

In this case, there are at least two chords in the $ba$-segment. It
follows that there are at most three internal vertices in the
$ab$-segment of $H$. Now we suppose there are three internal
vertices in the $ab$-segment of $H$, since it is easy to check that
$rc(G)\leq 6$ for the other two cases. If $ab\in E(G)$, then there
exists a connected dominating set with three vertices and then
$rc(G)\leq 5$. Otherwise, we claim that there exists a vertex $v$ in
the $ba$-segment such that $va,\ vb\in E(G)$, since $G$ has diameter
$3$ and at least two chords. It implies that $G$ has a connected
dominating set with four vertices, then $rc(G)\leq 6$.

{\bf Case 2.} There are some chords with both ends on the
$ab$-segment of $H$.

Choose a vertex cut of size two, $\{c,d\}$, such that any other
vertex cut of size two with both vertices in the $ab$-segment of $H$
has at least one vertex in the $cd$-segment of $H$, where the
$cd$-segment is a part of the $ab$-segment.

{\bf Subcase 2.1.} $a,b,c,d$ are not all distinct vertices.

Without loss of generality, we suppose $b=d$. By our choice, any
vertex on the $ac$-segment of $H$ does not form of a vertex cut with
$b$, and hence $ac$ must be an edge of $G$ ($ac$ may be an edge of
$H$ or a chord of $H$).

Suppose $\{c,b\}$ can not dominate the $cb$-segment. Let $v$ be a
vertex on the $vc$-segment such that $d(v,c)\geq 2$ and $d(v,b)\geq
2$. Then all vertices in $ba$-segment must be adjacent to vertex
$b$. Therefore, all vertices in $ac$-segment must be adjacent to
vertex $c$, since otherwise, if there exists a vertex $w$ such that
$wa\in E(G)$ and $wc\notin E(G)$, then $d(w,v)\geq 4$. Thus,
$\{b,c\}$ is a vertex cut with two vertices, which dominates more
vertices than $\{a,b\}$, a contradiction to the choice of $\{a,b\}$.

Now suppose $\{c,b\}$ dominates the $cb$-segment. Thus, $\{a,b,c\}$
must be a dominating set of $G$. If one of $ab$ and $bc$ is an edge
of $G$, then $\{a,b,c\}$ is a connected dominating set of $G$ and
thus $rc(G)\leq 2+3=5$. Now suppose neither $ab$ nor $bc$ is an edge
of $G$.

{\bf Subsubcase 2.1.1.} There is vertex $v$ in $ba$-segment (or
$cb$-segment) such that $v$ is adjacent to both $a$ and $b$ (or $c$
and $b$).

In this situation, $\{a,b,c,v\}$ is a connected dominating set of
$G$ and thus $rc(G)\leq 3+3=6$.

{\bf Subsubcase 2.1.2.} Otherwise, there does not exist such vertex.

Each vertex in $ba$-segment is only adjacent to one of $a$ and $b$,
and each vertex in $cb$-segment is only adjacent to one of $c$ and
$b$. Now in this case, each of $ba$-segment and $cb$-segment of $H$
has at least two internal vertices. We claim each of $ba$-segment
and $cb$-segment of $H$ has exactly two internal vertices, since
otherwise, we always can find two vertices with distance at least
four. Since $G$ has at least two chords, then we can assume that
$ac$-segment has at least two internal vertices, which also implies
one pair of vertices with distance at least four.

{\bf Subcase 2.2.} $a,b,c,d$ are distinct vertices.

The choice of $\{c,d\}$ implies that neither $ad$ nor $bc$ is an
edge of $G$. From the way that $\{a,b\}$ and $\{c,d\}$ was chosen,
we know that $\{a,b\}$ dominates the $ba$-segment and $\{c,d\}$
dominates the $dc$-segment. Moreover, $ac$ and $bd$ must be edges of
$G$. If there is one vertex $p$ in $ba$-segment such that it is
adjacent to $a$ but not adjacent to $b$, and also one vertex $q$ in
$cd$-segment such that it is adjacent to $d$ but not adjacent to
$c$, then $d(p,q)\geq 4$, a contradiction. Therefore, either
$\{a,b,c\}$ or $\{b,c,d\}$ is a dominating set of $G$. We assume
that $\{a,b,c\}$ is a dominating set of $G$, since the other case is
similar. If $ab$ is an edge of $G$, then $\{a,b,c\}$ is a connected
dominating set of $G$ and thus $rc(G)\leq 2+3=5$. Now suppose $ab$
is not an edge of $G$.

{\bf Subsubcase 2.2.1.} There is vertex $v$ in $ba$-segment such
that $v$ is adjacent to both $a$ and $b$.

In this situation, $\{a,b,c,v\}$ is a connected dominating set of
$G$ and thus $rc(G)\leq 3+3=6$.

{\bf Subsubcase 2.2.2.} Otherwise, there does not exist such vertex.

In this situation, each vertex in the $cd$-segment must be adjacent
to both $c$ and $d$, which implies that $cd$-segment contains
exactly one internal vertex. Similarly, there exactly two internal
vertices in the $ba$-segment. Since $G$ has two chords, then there
are some internal vertices in the $ac$-segment and $bd$-segment. In
each case, we can find two vertices with distance at least four.

The proof is thus completed.

\end{proof}

\section{Rainbow vertex-connection for line graphs}
In \cite{CLS}, the complexity of determining the rainbow
vertex-connection of a graph has been studied. The following result
was proved.
\begin{lem}[\cite{CLS}]\label{lem4}
The following problem is NP-Complete: given a vertex-colored graph
$G$, check whether the given coloring makes $G$ rainbow
vertex-connected.
\end{lem}

We will prove that it is still NP-Complete to decide whether a given
vertex-colored graph is rainbow connected even when the
vertex-colored graph is a line graph.

\begin{thm}
The following problem is NP-Complete: given a vertex-colored line
graph $G$, check whether the given coloring makes $G$ rainbow
vertex-connected.
\end{thm}
\begin{proof}
By Lemma \ref{lem1}, it will suffice by showing a polynomial
reduction from the problem in Lemma \ref{lem1}.

Given a graph $G=(V,E)$ and an edge-coloring $c$ of $G$. We want to
construct a line graph $G'$ with a vertex coloring such that $G'$ is
rainbow vertex-connected iff $G$ is rainbow connected.

Let $G=(V,E)$ and suppose $V=\{v_1,x_2,\ldots,v_n\}$ and
$E=\{e_1,e_2,\ldots,e_m\}$. Let $G_0=(V_0,E_0)$ be a new graph,
which is obtained from $G$ by attaching a pendent vertex $u_i$ to
$v_i$ for each $1\leq i\leq n$. Thus, $V_0=V\cup
\{u_1,u_2,\ldots,u_n\}$ and $E_0=E\cup \{e'_i=u_iv_i: 1\leq i\leq
n\}$. Let $G'$ be the line graph of $G_0$ and then $V(G')=E_0$. Now
we define a vertex coloring $c'$ as follows: for each $1\leq i\leq
n$, $c'(e_i)=c(e_i)$ and $c'(e'_i)=c_0$, where $c_0$ is a new color
we introduced.

Suppose $G$ is rainbow connected under the edge coloring $c$, then
we will check that there exists one vertex rainbow path between any
pair of vertices in $G'$ under the vertex coloring $c'$. Consider
the pair of $e'_i$ and $e'_j$ for $i\neq j$. Let
$v_{i_0}v_{i_1}\ldots v_{i_{k+1}}$ be the rainbow path between $v_i$
and $v_j$ in $G$, where $v_{i_0}=v_i$ and $v_{i_{k+1}}=v_j$. Denote
by $e_{i_t}=v_{i_t}v_{i_{t+1}}$ for $0\leq t\leq k$. Thus, we have
that edges $e_{i_0},\ e_{i_1},\ \ldots,\ e_{i_k}$ have distinct
colors. By the definition of $c'$, the colors of vertices $e_{i_0},\
e_{i_1},\ \ldots,\ e_{i_k}$ in $G'$ are all distinct. Thus,
$e'_ie_{i_0}e_{i_1}\ldots e_{i_k}e'_j$ is a required vertex rainbow
path. Similarly, for the pair $e'_i$ and $e_j$, and the pair $e_i$
and $e_j$, we can find vertex rainbow paths in $G'$, respectively.

Now suppose $G'$ is rainbow vertex-connected under the vertex
coloring $c'$, then we will check that there exists one rainbow path
between any pair of vertices in $G$ under the coloring $c$. For each
pair $e'_i$ and $e'_j$, where $1\leq i\neq j\leq n$, there exists
one vertex rainbow path $e'_ie_{i_0}e_{i_1}\ldots e_{i_{k+1}}e'_j$,
i.e., $e_{i_0},\ e_{i_1},\ \ldots,\ e_{i_{k+1}}$ has distinct
colors. Observe that in $G$, one of end vertices of $e_{i_0}$ is
$v_i$ and one of end vertices of $e_{i_{k+1}}$ is $v_j$. Thus, there
indeed exists one rainbow path connecting $v_i$ and $v_j$.

The proof is thus completed.

\end{proof}

\end{document}